\documentclass{article}

\usepackage[letterpaper,top=2cm,bottom=2cm,left=3cm,right=3cm,marginparwidth=1.75cm]{geometry}

\usepackage{graphicx}
\usepackage{makeidx}
\usepackage{amsmath}

\usepackage{amsthm}
\usepackage{amssymb}
\usepackage{amsfonts}
\usepackage{calligra}
\usepackage{subcaption}
\usepackage{multicol}
\usepackage{multirow}
\usepackage{algorithm}
\usepackage{algorithmic}
\usepackage{makecell}
\usepackage{pifont}

\newtheorem{theorem}{Theorem}[section]

\newtheorem{definition}{Definition}[section]

\newcommand\keywords[1]{\textbf{Keywords}: #1}

\usepackage[hidelinks]{hyperref}

\begin{document}
\title{Competitive Analysis for Online Fair Division under Multiple Fairness Notions \\ (Extended Abstract)}
\author{Tianqi CHEN\thanks{School of Mathematical Science, Zhejiang University, Hangzhou 310027, P. R. China.} \and Zhiyi TAN\thanks{Corresponding author. School of Mathematical Science, Zhejiang University,
		Hangzhou 310027, P. R. China. Supported by the National Natural
		Science Foundation of China (12571340). {\tt
			tanzy@zju.edu.cn}} }
\date{}
\maketitle

\maketitle

\begin{abstract}
We study the online fair division of indivisible items with additive utilities, where items arrive sequentially and must be irrevocably allocated upon arrival. Considering various fairness notions, we focus on designing online algorithms that produce fair or approximately fair allocations for any instance.
We measure algorithm performance using the competitive ratio, defined as the worst-case ratio between the fairness guarantee achieved by the online algorithm and that of an optimal offline allocation with full knowledge of future arrivals.
We examine a broad spectrum of models, including the allocation of goods or chores, normalized versus non-normalized utilities, and identical versus general utility functions. We address the majority of the unresolved cases by providing online algorithms or proving the limits of the competitive ratio achievable by online algorithms. In most cases, the algorithms are shown to be optimal. 
\end{abstract}

\keywords{Fair division, Online algorithms, Competitive ratio, Scheduling}

\section{Introduction}

Fair division studies how to allocate a set of $m$ indivisible items among $n$ agents in a fair manner.
Each agent is endowed with a utility function over subsets of items, and they all wish for the allocation to be relatively fair. A wide range of fairness notions has been proposed to evaluate allocations, including envy-freeness (EF)~\cite{Foley66}, proportionality (PROP)~\cite{Steihaus48}, and maximin share (MMS)~\cite{PW14}, as well as their relaxations, such as EF1 \cite{Budish11}, EFX \cite{CKMPSW19}, PROP1 \cite{CFS17}, and PROPX \cite{Moulin19}, among others.

Exact fairness guarantees are often unattainable with indivisible items, even for relaxations of classical fairness notions.  As a result, researchers have shifted their focus toward approximate fair division and established metrics quantifying the closeness to exact fairness~\cite{ABM18}.

While the majority of existing literature addresses fair division in the \emph{offline} setting, where all items are known in advance, many real-world applications involve items arriving sequentially and requiring immediate, irrevocable allocation decisions. This motivates the study of \emph{online fair division}, which has been actively investigated in recent years~\cite{ZBW23, NPT25, STW25}.

{\bf Our Contribution} In this paper, we study online fair division of indivisible items with additive utilities. Items arrive one by one. When allocating an item, only the utilities of the arrived items for all agents are known, while the utilities of future items are unknown. Depending on the nature of the items, we distinguish between \emph{goods}, which yield positive utilities (values), and \emph{chores}, which incur negative utilities (costs). We further consider both \emph{identical} utilities, where all agents share the same utility function, and \emph{general} utilities, where agents may have heterogeneous preferences. To better understand the role of information in online fair division, we consider the normalized setting, in which the total utility of all items is known in advance. If the total utility is not known a priori, the setting is referred to as non-normalized.

We design online algorithms to allocate items and measure their performance via the \emph{competitive ratio}, defined as the worst-case multiplicative gap between the fairness guarantee achieved by the online algorithm and that of an optimal offline allocation with full knowledge of future arrivals. This differs from some prior work on online fair division (e.g.,~\cite{ZBW23,NPT25}), which evaluates fairness guarantees without normalizing by the optimal offline benchmark.

To provide a more in-depth analysis of the properties of online fair division and their connections to other problems, we establish several structural properties of the problem, focusing on settings with identical utility functions or two agents. 
In particular, we clarify the connection between fair division and parallel machine scheduling, formalize the relationship between goods and chores allocation, and derive links between different fairness notions such as EF and PROP. 
These results provide a unifying perspective and serve as a foundation for our subsequent analysis.

Our main results are summarized in the following two tables, Table~\ref{tab:goods_results} and Table~\ref{tab:chores_results}. For goods allocation, we design online algorithms and derive tight competitive ratio bounds under various fairness notions, achieving optimal guarantees in all considered cases, with the exception of PROP1 under general utilities which remains largely unexplored. For chores allocation, we obtain competitive ratio bounds that are optimal in most cases. Compared to the goods setting, chores exhibit both similarities and notable differences across fairness notions. 
For instance, under EF, EF1, and EFX, the competitive ratios in goods and chores allocation are closely related and often reciprocal. 
On the other hand, chores can be strictly easier in certain cases. Under identical utilities, PROPX can be achieved for any number of agents $n$ in the normalized setting, whereas in the goods allocation this is only possible when $n=2$. 
At the same time, several cases for chores remain open or not fully resolved, particularly under general utilities and for fairness notions such as PROP and PROPX.

\begin{table}[t]
	\renewcommand{\arraystretch}{1.5} 
	\centering
	\begin{tabular}{|l|c|c|cc|cc|}
		\hline \multicolumn{1}{|c|}{\multirow{3}{*}{\makecell{Fairness\\ notions}}} 
		&
		\multicolumn{2}{c|}{Non-normalized} & 
		\multicolumn{4}{c|}{Normalized}
		\\
		\cline{2-7}
		& \multirow{2}{*}{Identical}
		& \multirow{2}{*}{General}
		& \multicolumn{2}{c|}{Identical} 
		& \multicolumn{2}{c|}{General}  \\ \cline{4-7}
		\multicolumn{1}{|c|}{}   
		& 
		&                                  
		& \multicolumn{1}{c|}{$n=2$} 
		& $n\geq 3$ 
		& \multicolumn{1}{c|}{$n=2$} 
		& $n\geq 3$ \\ \hline
		
		EF  & $\frac{1}{n+1}^{(*)}$
		& $0^{(*)}$
		& \multicolumn{2}{c|}{$\frac{1}{n}^{(*)}$}
		& \multicolumn{1}{c|}{$\frac{1}{3}^{(*)}$} 
		& $0^{(*)}$ \\ \hline
		
		EFX & $0^{\text{\cite{NPT25}}}$
		& $0^{\text{\cite{NPT25}}}$
		& \multicolumn{1}{c|}{$\frac{\sqrt{5}-1}{2}^{\text{\cite{NPT25,MP25}}}$} 
		& $0^{\text{\cite{NPT25,MP25}}}$
		& \multicolumn{2}{c|}{$0^{\text{\cite{NPT25,MP25}}}$}  
		\\ \hline
		
		EF1 & $1^{\text{\cite{ELL25}}}$
		& $0^{\text{\cite{NPT25}}}$
		& \multicolumn{2}{c|}{$1^{\text{\cite{ELL25}}}$} 
		& \multicolumn{1}{c|}{$1^{\text{\cite{NPT25}}}$} 
		& $0^{\text{\cite{NPT25}}}$
		\\ \hline
		
		PROP & $\frac1n^{\text{\cite{Woeginger97}}}$
		& $0^{(*)}$
		& \multicolumn{1}{c|}{$\frac{2}{3}^{\text{\cite{KKST97}}}$} 
		& $\frac{1}{n-1}^{\text{\cite{TW07}}}$
		& \multicolumn{1}{c|}{$\frac{1}{2}^{(*)}$} 
		& $0^{(*)}$
		\\ \hline
		
		PROPX & $0^{(*)}$
		& $0^{(*)}$
		& \multicolumn{1}{c|}{$1^{(*)}$} 
		& $0^{(*)}$
		& \multicolumn{2}{c|}{$0^{(*)}$}
		\\ \hline
		
		PROP1 & $1^{\text{\cite{ELL25}}}$
		& $<1 ^{\text{\cite{BKP18}}}$
		& \multicolumn{2}{c|}{$1^{\text{\cite{ELL25}}}$ }
		& \multicolumn{2}{c|}{$1^{\text{\cite{NPT25}}}$} 
		\\ \hline
		
		MMS & $\frac1n^{\text{\cite{Woeginger97}}}$
		& $0^{\text{\cite{ZBW23}}}$
		& \multicolumn{1}{c|}{$\frac{2}{3}^{\text{\cite{KKST97}}}$} 
		& $\frac{1}{n-1}^{\text{\cite{TW07}}}$
		& \multicolumn{1}{c|}{$\frac{1}{2}^{\text{\cite{ZBW23}}}$} 
		& $0^{\text{\cite{ZBW23}}}$ 
		\\ \hline
		
	\end{tabular}
	
		\caption{Competitive ratios of online algorithms for goods allocation under various fairness notions.
			Entries superscripted by $(*)$ refer to results proved in this paper,
			while entries superscripted by $[\cdot]$ refer to prior work.
		}
		\label{tab:goods_results}
	\end{table}
	
	\begin{table}[!htbp]
		\renewcommand{\arraystretch}{1.5} 
		\centering
		\begin{tabular}{|l|c|c|cc|cc|}
			\hline \multicolumn{1}{|c|}{\multirow{3}{*}{\makecell{Fairness\\ notions}}}
			&
			\multicolumn{2}{c|}{Non-normalized}
			& 
			\multicolumn{4}{c|}{Normalized}
			\\
			\cline{2-7}
			& \multirow{2}{*}{Identical}
			& \multirow{2}{*}{General} 
			& \multicolumn{2}{c|}{Identical} 
			& \multicolumn{2}{c|}{General} \\ \cline{4-7}
			\multicolumn{1}{|c|}{}       
			& 
			& 
			& \multicolumn{1}{c|}{$n=2$} 
			& $n\geq 3$ 
			& \multicolumn{1}{c|}{$n=2$} 
			& $n\geq 3$ \\ \hline
			
			EF  & $n+1^{(*)}$
			& $\infty^{(*)}$
			& \multicolumn{2}{c|}{$n^{(*)}$ } 
			& \multicolumn{1}{c|}{$3^{(*)}$ } 
			& $\infty^{(*)}$  \\ \hline
			
			EFX & $\infty^{(\ref{thm:goods-chores-reciprocal}), \ \text{\cite{NPT25}}}$ 
			& $\infty^{(*)}$
			& \multicolumn{1}{c|}{$\frac{\sqrt{5}+1}{2}^{(\ref{thm:goods-chores-reciprocal}), \ \text{\cite{NPT25,MP25}}}$ } 
			& $\infty^{(\ref{thm:goods-chores-reciprocal}), \ \text{\cite{NPT25,MP25}}}$
			& \multicolumn{2}{c|}{$\infty^{(*)}$ }
			\\ \hline
			
			EF1 & $1^{(\ref{thm:goods-chores-reciprocal}), \ \text{\cite{ELL25}}}$
			& $\infty^{(\ref{thm:chores_online_general_ef}, \ \ref{thm:chores_semi-online_general_ef1_ngeq3})}$
			& \multicolumn{2}{c|}{$1^{(\ref{thm:goods-chores-reciprocal}), \ \text{\cite{ELL25}}}$} 
			& \multicolumn{1}{c|}{$1^{(*)}$} 
			& $\infty^{(*)}$
			\\ \hline
			
			PROP & $[1.85358^{\text{\cite{GRT00}}}, 1.9201^{\text{\cite{FW00}}}]$ 
			& $n^{(*)}$
			& \multicolumn{1}{c|}{$\frac{4}{3}^{\text{\cite{KKST97}}}$ } 
			& $1.585^{\text{\cite{AH12,KKG15}}}$ 
			& \multicolumn{1}{c|}{$\frac{3}{2}^{(*)}$} 
			& $[2,n]^{(*)}$
			\\ \hline
			
			PROPX & $n^{(*)}$
			& $n^{(*)}$
			& \multicolumn{2}{c|}{$1^{(*)}$} 
			& \multicolumn{2}{c|}{$[\frac{n}{n-1},n]^{(*)}$}  
			\\ \hline
			
			PROP1 & $1^{(\ref{thm:goods-chores-reciprocal}), \ \text{\cite{ELL25}}}$ 
			& $>1 ^{\text{\cite{BKP18}}}$
			& \multicolumn{2}{c|}{$1^{(\ref{thm:goods-chores-reciprocal}), \ \text{\cite{ELL25}}}$} 
			& \multicolumn{2}{c|}{$1^{(*)}$} 
			\\ \hline
			
			MMS & $[1.85358^{\text{\cite{GRT00}}}, 1.9201^{\text{\cite{FW00}}}]$ 
			& $n^{\text{\cite{STW25}}}$ 
			& \multicolumn{1}{c|}{${\frac43}^{\text{\cite{KKST97}}}$ }
			& $1.585^{\text{\cite{AH12,KKG15}}}$ 
			& \multicolumn{1}{c|}{$[\frac{15}{11},\sqrt{2}]^{\text{\cite{ZBW23}}}$} 
			& $[1.585^{\text{\cite{AH12}}},2-\frac{1}{n}^{\text{\cite{ZBW23}}} ]$ 
			\\ \hline
			
		\end{tabular}
		
		\caption{Competitive ratios of online algorithms for chores allocation under various fairness notions.
			Entries that are not tight are shown in brackets, indicating the best known lower and upper bounds.
			Results enclosed in parentheses $(*)$ are established in this paper; results enclosed in square brackets $[\cdot]$ are from prior work.
		}
		\label{tab:chores_results}
	\end{table}

{\bf Related Works} Early work on online fair division primarily focused on qualitative fairness properties or efficiency objectives, rather than on approximating fairness guarantee.
Aleksandrov et al.~\cite{AAG15} were among the first to study online fair division, proposing a model motivated by food bank allocation that emphasizes axiomatic fairness properties such as envy-freeness and strategy-proofness, together with efficiency measures like competitive ratio and the price of anarchy.
Benade et al.~\cite{BKP18} considered online fair division with indivisible goods under
sequential arrivals and evaluated performance via the maximum envy accumulated
over time, designing algorithms with asymptotically vanishing envy.
Relatedly, Banerjee et al.~\cite{BGG22} studied the online allocation of sequentially arriving $T$ divisible goods with the social objectives of maxmin and Nash social welfare, respectively, assuming that each agent’s total utility over all goods is predictable.
For a broader overview of early work on online fair division, we refer the
reader to the survey~\cite{AW20}.

A more recent line of work studies online fair division through the lens of
approximate fairness guarantees and competitive analysis.
This approach was initiated by Zhou et al.~\cite{ZBW23}, who analyzed online fair division with respect to the MMS guarantee.
Building on this line of research, Neoh et al.~\cite{NPT25} established tight competitive ratios of goods for several fairness notions, including EF1, EFX, and PROP1.
In the context of chores, Song et al.~\cite{STW25} derived
improved lower bounds for MMS guarantees.

Several recent studies have explored extensions and variations of the aforementioned models. For instance, Choo et al. \cite{CFKNPT26} investigated online goods allocation where the maximum value of any item for each agent is known in advance. 
They proved that no online algorithm can achieve a strictly positive approximation guarantee with respect to EF1, PROPX, or MMS. In contrast, for PROP1, an algorithm achieving an approximation guarantee of at least $\frac{1}{n}$ exists. Furthermore, even when the known maximum value is only a prediction, the resulting approximation guarantee can be expressed as a function of the prediction error. 
They also studied the performance of randomized algorithms under the PROP1 notion.
Melissourgos and Protopapa~\cite{MP25} investigated online EFX allocations incorporating predictions of future item values. Wang and Wei~\cite{WW26} addressed the online allocation of goods or chores under personalized bi-valued utilities, which are not necessarily additive. Furthermore, Amanatidis et al.~\cite{ALM25} examined online goods allocation with personalized bi-valued utilities, evaluating the approximate fairness guarantee based on the worst value achieved at any point during the sequential allocation process.

Fair division is closely related to scheduling problems,
especially under additive and identical utilities.
Motivated by this connection, we briefly review relevant results
from online and semi-online scheduling.

Online scheduling on parallel machines has been extensively studied,
with two fundamental objective functions:
minimizing the maximum load and maximizing the minimum load.
The former objective, commonly known as minimizing the makespan,
was initiated by Graham~\cite{Graham66}, who designed the famous List Scheduling (LS) algorithm, which is the optimal algorithms for two and three machines \cite{FKT89}.
For four or more machines, however, the optimal competitive ratio
remains open, despite a long line of work providing progressively improved bounds
(e.g.,~\cite{CVW94,FW00,GRT00,RC03}).
The latter objective, maximizing the minimum load, was first studied by
\cite{DFL82} and was later shown by \cite{Woeginger97}
to admit an optimal competitive ratio for an arbitrary number of machines.

Beyond the fully online model, a rich line of work studies
\emph{semi-online} scheduling,
where limited global information about the instance is revealed in advance.
Such models interpolate between offline and online scheduling
and help quantify the value of partial information;
see the surveys~\cite{TZ13,Epstein18}.
The online fair allocation problem in the normalized setting studied in this paper is closely related to semi-online scheduling with total processing time is known in advance \cite{KKST97, AH12, KKG15, TW07}. The assumption in \cite{CFKNPT26} that each agent's maximum item value is known or predictable extends the semi-online scheduling framework where the maximum job processing time is either fully known or bounded within an interval \cite{TW07, GLT25, TH07}.

\section{Preliminaries}
\label{sec:prelim}

\subsection{Agents and Items}

We investigate a fair division problem with a set \( M = \{e_1,e_2,\dots,e_m\} \) of \( m \) indivisible items and a set \( N \) of \( n \) agents.
Each agent \( i \in N \) has a non-negative \emph{additive utility function} $u_i:2^{M}\rightarrow \mathbb{R}_{\geq 0}$, defined by
\(
u_i(S) = \sum_{e \in S} u_i(\{e\})
\)
for any subset \( S \subseteq M \).
For convenience, we abuse notation and write $e$ for the singleton set $\{e\}$.
We consider two cases based on whether the utility functions of different agents are identical. In the \emph{identical} case, all agents share the same utility function,
i.e., \( u_i = u_j \) for all \( i,j \in N \).
In the \emph{general} case, agents may have different utility functions.
Alternatively, we classify the problem based on whether the total utility for each agent is the same. In the \emph{normalized} case, \( u_i(M)\) is constant for all \( i \in N \). We assume w.l.o.g. that \( u_i(M) = n \) for all \( i \in N \). In the \emph{non-normalized} case,  the values of \( u_i(M)\) may differ across agents.

We consider two types of items.
For \emph{goods}, $u_i(e)$ denotes the \emph{value} of item $e$ to agent $i$, and we write $v_i$ for the \emph{valuation function}.
For \emph{chores}, $u_i(e)$ denotes the \emph{cost} incurred by agent $i$ for item $e$, and we write $c_i$ for the \emph{cost function}.
Under identical utilities, we omit the subscript $i$.

In the online allocation model, items arrive sequentially. Upon the arrival of each item, it must be irrevocably assigned to one of the agents based only on the information revealed so far. Reallocation of previously assigned items is not allowed. However, for the normalized case, the total utility $n$ is known in advance and can be regarded as additional information.

An $n$-partition of \( M \) is a set of $n$ bundles \( {\cal A} = (A_1, \dots, A_n) \) such that $\cup_{i=1}^n A_i=M$ and $A_i\cap A_j =\emptyset $ for $i\neq j$. The set of all $n$-partitions of $M$ is denoted $\Pi_n(M)$. An allocation corresponds to a $n$-partition \( {\cal A} \in \Pi_n(M) \) where agent \( i \) receives bundle \( A_i \).
Let \( A_i^j \) denote the bundle held by agent \( i \) after the first \( j \) items have been allocated.
In particular, \( A_i^0 = \emptyset \) and \( A_i^m = A_i \).

\subsection{Fairness Notions}

We now formally define the fairness notions considered in this paper,
namely EF, EFX, EF1, PROP, PROPX, PROP1, and MMS,
for both goods and chores.

\begin{definition}
	[Envy-freeness]
	An allocation ${\cal A}$ is called \emph{$\alpha$-approximate envy-free ($\alpha$-EF)}
	if the following holds.
	For \emph{goods}, $0\le\alpha\leq 1$ and for all $i,j\in N$,
	\[
	v_i(A_i)\ge \alpha\, v_i(A_j),
	\]
	while for \emph{chores}, $\alpha\ge 1$ and for all $i,j\in N$,
	\[
	c_i(A_i)\le \alpha\, c_i(A_j).
	\]
	When $\alpha=1$, the allocation is called EF.
\end{definition}

\begin{definition}
	[Envy-freeness up to Any Item]
	An allocation ${\cal A}$ is called \emph{$\alpha$-approximate envy-free up to any item ($\alpha$-EFX)}
	if the following holds.
	For \emph{goods}, $0\le\alpha\leq 1$ and for all $i,j\in N$ and any $e\in A_j$,
	\[
	v_i(A_i)\ge \alpha\, v_i(A_j\setminus e),
	\]
	while for \emph{chores}, $\alpha\ge 1$ and for all $i,j\in N$ and any $e\in A_i$,
	\[
	c_i(A_i\setminus  e)\le \alpha\, c_i(A_j).
	\]
	When $\alpha=1$, the allocation is called \emph{EFX}.
\end{definition}

\begin{definition}
	[Envy-freeness up to One Item]
	An allocation ${\cal A}$ is called \emph{$\alpha$-approximate envy-free up to one item ($\alpha$-EF1)}
	if the following holds.
	For \emph{goods}, $0\le\alpha\leq 1$ and for all $i,j\in N$,
	there exists $e\in A_j$ such that
	\[
	v_i(A_i)\ge \alpha\, v_i(A_j\setminus e),
	\]
	while for \emph{chores}, $\alpha\ge 1$ and for all $i,j\in N$,
	there exists $e\in A_i$ such that
	\[
	c_i(A_i\setminus  e)\le \alpha\, c_i(A_j).
	\]
	When $\alpha=1$, the allocation is called \emph{EF1}.
\end{definition}

\begin{definition}
	[Proportionality]
	An allocation ${\cal A}$ is called \emph{$\alpha$-approximate proportional ($\alpha$-PROP)}
	if the following holds.
	For \emph{goods}, $0\le\alpha\leq 1$ and for all $i\in N$,
	\[
	v_i(A_i)\ge \frac{\alpha}{n}\, v_i(M),
	\]
	while for \emph{chores}, $\alpha\ge 1$ and for all $i\in N$,
	\[
	c_i(A_i)\le \frac{\alpha}{n}\, c_i(M).
	\]
	When $\alpha=1$, the allocation is called \emph{PROP}.
\end{definition}

\begin{definition}
	[Proportionality up to Any Item]
	An allocation ${\cal A}$ is called \emph{$\alpha$-approximate proportional up to any item ($\alpha$-PROPX)}
	if the following holds.
	For \emph{goods}, $0\le\alpha\leq 1$ and for all $i\in N$ and any $e\in M\setminus A_i$,
	\[
	v_i(A_i\cup e)\ge \frac{\alpha}{n}\, v_i(M),
	\]
	while for \emph{chores}, $\alpha\ge 1$ and for all $i\in N$ and any $e\in A_i$,
	\[
	c_i(A_i\setminus e)\le \frac{\alpha}{n}\, c_i(M).
	\]
	When $\alpha=1$, the allocation is called \emph{PROPX}.
\end{definition}

\begin{definition}
	[Proportionality up to One Item]
	An allocation ${\cal A}$ is called \emph{$\alpha$-approximate proportional up to one item ($\alpha$-PROP1)}
	if the following holds.
	For \emph{goods}, $0\le\alpha\leq 1$ and for all $i\in N$,
	there exists $e\in M\setminus A_i$ such that
	\[
	v_i(A_i\cup e)\ge \frac{\alpha}{n}\, v_i(M),
	\]
	while for \emph{chores}, $\alpha\ge 1$ and for all $i\in N$,
	there exists $e\in A_i$ such that
	\[
	c_i(A_i\setminus e)\le \frac{\alpha}{n}\, c_i(M).
	\]
	When $\alpha=1$, the allocation is called \emph{PROP1}.
\end{definition}

\begin{definition}
	[Maximin Share Fairness]
	An allocation ${\cal A}$ is called \emph{$\alpha$-approximate maximin share fair ($\alpha$-MMS)} if the following holds.
	For \emph{goods}, $0\le\alpha\leq 1$ and for all $i\in N$,
	\[ v_i(A_i)\ge \alpha\, \mathrm{MMS}_i,\]
    where 
    \[
    \mathrm{MMS}_i
    =
    \max_{{\cal X}\in\Pi_n(M)} \min_{j\in N} v_i(X_j)
    \]
    is the \emph{maximin share} of agent $i$.
    For \emph{chores}, $\alpha\ge 1$ and for all $i\in N$,
    \[
	c_i(A_i)\le \alpha\, \mathrm{MMS}_i,
	\]
	where
	\[
    \mathrm{MMS}_i=\min_{{\cal X}\in\Pi_n(M)} \max_{j\in N} c_i(X_j)
    \]
    is the \emph{maximin share} of agent $i$.
	When $\alpha=1$, the allocation is called \emph{MMS}.
\end{definition}

From the above definitions, it can be seen that if an allocation is 
$\alpha$-EF (resp. $\alpha$-PROP), then it must also be $\alpha$-EFX (resp. $\alpha$-PROPX), but the converse does not hold. Similarly, if an allocation is 
$\alpha$-EFX (resp. $\alpha$-PROPX), then it must also be $\alpha$-EF1 (resp. $\alpha$-PROP1), but the converse does not hold. Moreover, if an allocation is 
$\alpha$-EF (resp. $\alpha$-EF1), then it must also be $\alpha$-PROP (resp. $\alpha$-PROP1), but the converse does not hold. However, the implication that an 
$\alpha$-EFX allocation must also be $\alpha$-PROPX holds only for chores allocation.
More intricate relationships among different notions of approximately fair allocations can be found in \cite{ABM18}.

\subsection{Fairness Approximation Guarantee and Competitive Ratio}
\label{sec:competitive_ratio}

We now formalize the notions of approximation guarantees with respect to fairness and the competitive ratio used in this paper.
Our definition refines the standard competitive analysis framework to accommodate approximate fairness objectives, making precise the comparison between online allocations and optimal offline benchmarks.

Given an instance \( I \) of goods allocation, let 
\( \mathrm{F}_{\mathrm{g}}^{\cal A}(I) = \alpha \), where 
\( \alpha \) is the maximum value such that allocation 
${\cal A}$ is \( \alpha \)-approximate with respect to fairness notion F. Similarly, for an instance $I$ of chores allocation, we define 
\( \mathrm{F}_{\mathrm{c}}^{\cal A}(I) = \alpha \), where 
\( \alpha \) is the minimum value such that ${\cal A}$ is \( \alpha \)-approximate with respect to F.
We refer to \( \alpha \) as the \emph{approximation guarantee} of allocation \( {\cal A}  \)
with respect to \( \mathrm{F} \). When the context is clear, it is abbreviated as \( \mathrm{F}^{\cal A}(I) = \alpha \). 
The expressions of approximation guarantees under different fairness notions are summarized in Table \ref{table4}. 

\begin{table}[htbp]
	\centering
	\begin{tabular}{|c|c|c|}
\hline
		Fairness & Goods & Chores \\
\hline
		$\mathrm{EF}^{\cal A}(I)$ & $\min\left\{1, \min\limits_{i,j\in N} \frac{v_i(A_i)}{ v_i(A_j)}\right\}$ & $\max\left\{1, \max\limits_{i,j\in N} \frac{c_i(A_i)}{c_i(A_j)} \right\}$ \\
\hline
$\mathrm{EFX}^{\cal A}(I)$ & $\min\left\{1,  \min\limits_{i,j\in N}\frac{ v_i(A_i)}{v_i(A_j)-\min\limits_{e\in A_j} v_i(e)}\right\}$ & $\max\left\{1, \max\limits_{i,j\in N}\frac{c_i(A_j)-\min\limits_{e\in A_j} c_i(e)}{ c_i(A_i)} \right\}$		 \\
\hline
$\mathrm{EF1}^{\cal A}(I)$ & $\min\left\{1,  \min\limits_{i,j\in N}\frac{ v_i(A_i)}{v_i(A_j)-\max\limits_{e\in A_j} v_i(e)}\right\}$ & $\max\left\{1, \max\limits_{i,j\in N}\frac{ c_i(A_j)-\max\limits_{e\in A_j} c_i(e)}{ c_i(A_i)} \right\}$		 \\
\hline		
$\mathrm{PROP}^{\cal A}(I)$ & $\min\left\{1,\min\limits_{i\in N} \frac{nv_i(A_i)}{v_i(M)}\right\}$ & $\max\left\{1,\max\limits_{i\in N} \frac{nc_i(A_i)}{c_i(M)} \right\}$ \\
\hline
$\mathrm{PROPX}^{\cal A}(I)$ & $\min\left\{1,\min\limits_{i\in N} \frac{n\left( v_i(A_i)+\min\limits_{e\in M\setminus A_i}v_i(e)\right) }{v_i(M)}\right\}$ & $\max\left\{1,\max\limits_{i\in N} \frac{n\left( c_i(A_i)-\min\limits_{e\in A_i}c_i(e)\right) }{c_i(M)} \right\}$ \\
\hline
$\mathrm{PROP1}^{\cal A}(I)$ & $\min\left\{1,\min\limits_{i\in N} \frac{n\left( v_i(A_i)+\max\limits_{e\in M\setminus A_i}v_i(e)\right) }{v_i(M)}\right\}$ & $\max\left\{1,\max\limits_{i\in N} \frac{n\left( c_i(A_i)-\max\limits_{e\in A_i}c_i(e)\right) }{c_i(M)} \right\}$ \\
\hline
$\mathrm{MMS}^{\cal A}(I)$ & $\min\limits_{i\in N}\frac{v_i(A_i)}{\max\limits_{{\cal X}\in\Pi_n(M)} \min\limits_{j\in N} v_i(X_j)}$ & $\max\limits_{i\in N}\frac{c_i(A_i)}{\min\limits_{{\cal X}\in\Pi_n(M)} \max\limits_{j\in N} c_i(X_j)}$ \\
\hline
	\end{tabular}
	\caption{The expressions of approximation guarantee under different fairness notions}	\label{table4}
\end{table}

For each instance \( I \), let $\mathrm{F}^*(I)$ be the optimal approximation guarantee among all allocations with respect to \( \mathrm{F} \). That is, 
$$
\mathrm{F}_{\mathrm{g}}^*(I) = \sup_{ {\cal A} \in\Pi_n(M)} \mathrm{F}_{\mathrm{g}}^{\cal A}(I)
, \quad 
\mathrm{F}_{\mathrm{c}}^*(I) = \inf_{ {\cal A} \in\Pi_n(M)} \mathrm{F}_{\mathrm{c}}^{\cal A}(I).
$$
Clearly, $\mathrm{F}^*(I)=1$ implies that there exists an allocation for instance $I$
that exactly satisfies the fairness notion \( \mathrm{F} \).
For goods allocation, $\mathrm{F}_{\mathrm{g}}^*(I) < 1$ means that exact fairness is unattainable and only approximate fairness can be achieved, whereas for chores allocation, this corresponds to $\mathrm{F}_{\mathrm{c}}^*(I) > 1$.
Allocations whose approximation guarantee are exactly  $\mathrm{F}^*(I)$ are called optimal allocations with respect to F. For the same instance, optimal allocations under different fairness notions are not necessarily the same. For a given instance $I$ of goods allocation, $\mathrm{F}^*(I)=0$ indicates that there exists no non-trivial allocation under the fairness notion $\mathrm{F}$. For chore allocation, such instances are characterized by $\mathrm{F}^*(I)=\infty$ when $\mathrm{F}\in\{\mathrm{EF},\mathrm{EFX},\mathrm{EF1}\}$ and $\mathrm{F}^*(I)=n$ when $\mathrm{F}\in\{\mathrm{PROP},\mathrm{PROPX},\mathrm{PROP1}, \mathrm{MMS}\}$.

The optimal fairness approximation guarantee varies across instances. For a fixed fairness notion, finding the worst optimal approximation guarantee among all instances is a hot research topic in the field of fair division. Table \ref{tab:fairness_summary} summarizes the the best known results so far.

\begin{table}[h!]
\centering
\begin{tabular}{|c|c|c|c|c|}
\hline
\multirow{2}{*}{Fairness} & \multicolumn{2}{c|}{Goods} & \multicolumn{2}{c|}{Chores} \\
\cline{2-5}
 & Identical & General & Identical & General \\
\hline
EF    & 0 & 0 & 0 & 0 \\
\hline
EFX   & 1 & \makecell{$(n\le 3)$ 1 \cite{CGM24} \\ $(4\le n \le 7)$ $\frac23$ \cite{AFS24} \\  $(n\ge 8)$ $\frac{\sqrt{5}-1}{2}$ \cite{AMN20}} & 1 & \makecell{$(n = 2)$ 1\\ $(n\ge 3)$ $4$ \cite{GMQ25}} \\
\hline
EF1   & 1 & 1 \cite{LMM04} & 1 & 1 \cite{LMM04} \\
\hline
PROP  & 0 & 0 & 0 & 0 \\
\hline
PROPX & \makecell{$(n = 2)$ 1 \\ $(n\ge 3)$ 0} & \makecell{$(n = 2)$ 1\\ $(n\ge 3)$ 0} & 1 & 1 \cite{ALM24} \\
\hline
PROP1 & 1 & 1 \cite{LMM04} & 1 & 1 \cite{LMM04} \\
\hline
MMS   & 1 & \makecell{$(n = 2)$ 1 \cite{PW14}\\ $(n = 3)$ $\frac{11}{12}$ \cite{FN22} \\ $(n = 4)$ $\frac{4}{5}$ \cite{GHSSY18}\\ $(n\ge 5)$ $\frac{3}{4}+\frac{3}{3836}$ \cite{AG24}}  & 1 & \makecell{$(n= 2)$ 1 \cite{PW14} \\ $(n = 3)$ $\frac{15}{13} $ \cite{HS23}\\ $(4\le n \le 7)$ $\frac{20}{17}$ \cite{HS23} \\ $(n\ge 8)$ $\frac{13}{11}$ \cite{HS23}} \\
\hline
\end{tabular}
\caption{Best known approximation guarantees for different fairness notions. Some folklore results are presented without citation. }
\label{tab:fairness_summary}
\end{table}

In this paper, we focus on designing online algorithms to obtain allocations for any instance with the goal of optimizing the fairness approximation guarantee. The performance of an online algorithm is measured by its competitive ratio. Let \( H \) be an online allocation algorithm, and the allocation generated by \( H \)  is denoted as ${\cal H}$. We now give the formal definition of the competitive ratio. For instances with \( \mathrm{F}_{\mathrm{g}}^*(I) = 0 \) or \( \mathrm{F}_{\mathrm{c}}^*(I) = \infty\), it must be that \( \mathrm{F}_{\mathrm{g}}^{\cal H}(I) = 0 \) or  \( \mathrm{F}_{\mathrm{c}}^{\cal H}(I) = \infty\), and we conventionally set $\frac{\mathrm{F}^{\cal H}(I)}{\mathrm{F}^*(I)}=1$. 

\begin{definition}[Competitive Ratio]
The \emph{competitive ratio} of algorithm \( H \) with respect to a fairness notion \( \mathrm{F} \) is defined as
\[
\mathrm{CR}_{\mathrm{g}}^{\mathrm{F}}(H) = \inf_I \frac{\mathrm{F}_{\mathrm{g}}^{\cal H}(I)}{\mathrm{F}_{\mathrm{g}}^*(I)}
\]
for goods allocation, and 
\[
\mathrm{CR}_{\mathrm{c}}^{\mathrm{F}}(H) = \sup_I \frac{\mathrm{F}_{\mathrm{c}}^{\cal H}(I)}{\mathrm{F}_{\mathrm{c}}^*(I)}
\]
for chores allocation. 	
\end{definition}

When the context is clear, we use $\text{CR}^{\mathrm{F}}(H)$  
to denote the competitive ratio of algorithm $H$.
The closer the competitive ratio of an algorithm is to 1, the better its performance. If no online algorithm exists with a competitive ratio better than $r$, and the competitive ratio of \( H \) is exactly $r$, then \( H \) is referred to as \emph{optimal}. Obviously, if under a certain fairness notion, no algorithm achieves a competitive ratio better than $r$ with identical utility functions, then no such algorithm exists with general utility functions either. Similarly, if no algorithm achieves a competitive ratio better than 
$r$ in the normalized setting, then none exists in the non-normalized setting either.

Zhou et al. \cite{ZBW23} studied online allocation under the fair notion of MMS, and more fair notions such as EF1, EFX, and PROP1 were investigated  in \cite{NPT25}. They also use competitive ratio to measure performance of online algorithms, with a definition that  differs slightly from ours. They directly take the worst of the approximation guarantee achievable by the algorithm over all instances, i.e., $\inf\limits_I \mathrm{F}_{\mathrm{g}}^{\cal H}(I)$ or $\sup\limits_I \mathrm{F}_{\mathrm{c}}^{\cal H}(I)$, 
as its competitive ratio. In fact, if under a specific fairness notion F, $\mathrm{F}^*(I)=1$ for any instance $I$, the two definitions are equivalent. 
We have listed corresponding results of \cite{ZBW23} and \cite{NPT25} in Tables \ref{tab:goods_results} and \ref{tab:chores_results}.  

However, there are indeed other cases where $\mathrm{F}^*(I)\neq 1$ for some instance $I$ (Ref. Table \ref{tab:fairness_summary}). 
For these cases, the optimal approximation guarantee may vary across different instances under a specific fairness notion, let alone across different notions.
In a more extreme case, if there exists an instance $I$ of goods allocation such that under a specific fairness notion F, $\mathrm{F}^*(I)=0$, then the worst-case approximation guarantee of any algorithm can only be $0$. Since this instance admits no non-trivial allocation, one cannot expect any algorithm to achieve this.
Hence, measuring and comparing algorithm performance based on the approximation guarantee seems somewhat unreasonable. For this reason, in this paper, we adopt the standard approach commonly used in competitive analysis to define the competitive ratio of an online algorithm.

\section{Special Properties with Identical Utility Functions or Two Agents}
\label{sec:structural}

\subsection{Connection Between Fair Division and Parallel Machine Scheduling}

Parallel machine scheduling is also an optimization problem related to resource allocation, and thus it is closely connected to fair division, particularly in the case of identical utilities. 
By regarding agents as machines and items as jobs, the utility of an item to an agent corresponds to the processing time of the job on that machine. Each allocation corresponds to a feasible schedule. The total utility of the items obtained by an agent corresponds to the completion time of the machine. Therefore, the maximum and minimum utilities among all agents correspond to the makespan and the minimum completion time of all machines, respectively.

Specifically, in the case of identical utilities, the utility of any set of items is the same for all agents. 
A similar connection between MMS and parallel machine scheduling in this setting was previously observed by Zhou et al.~\cite{ZBW23}. 
For completeness, we provide a formal treatment, where a similar connection also extends to the refined definition of the competitive ratio under PROP. 
Let $u_{\max}^{\cal A}=\max_{i\in N}u(A_i)$ and $u_{\min}^{\cal A}=\min_{i\in N}u(A_i)$. Then 
$$\mathrm{EF}_{\mathrm{g}}^{\cal A}(I)= \frac{v_{\min}^{\cal A}}{v_{\max}^{\cal A}}, \, \mathrm{EF}_{\mathrm{c}}^{\cal A}(I)= \frac{c_{\max}^{\cal A}}{c_{\min}^{\cal A}}$$ 
and 
$$\mathrm{PROP}_{\mathrm{g}}^{\cal A}(I)=\frac{nv_{\min}^{\cal A}}{v(M)} , \, \mathrm{PROP}_{\mathrm{c}}^{\cal A}(I)=\frac{nc_{\max}^{\cal A}}{c(M)}.$$ 
Furthermore, for goods allocation and chores allocation, the maximin share of each agent equals to the optimal minimum machine completion time and optimal makespan, respectively.  Therefore, for any allocation ${\cal A}$, 
\[
\frac{\mathrm{PROP}_{\mathrm{g}}^{\cal A}(I)}{\mathrm{PROP}_{\mathrm{g}}^*(I)} = \frac{v_{\min}^{\cal A}}{\max_{{\cal X}\in \Pi_n(M)}v_{\min}^{\cal X}} = \mathrm{MMS}_{\mathrm{g}}^{\cal A}(I),
\]
and 
\[
\frac{\mathrm{PROP}_{\mathrm{c}}^{\cal A}(I)}{\mathrm{PROP}_{\mathrm{c}}^*(I)} = \frac{c_{\max}^{\cal A}}{\min_{{\cal X}\in \Pi_n(M)}c_{\max}^{\cal X}} = \mathrm{MMS}_{\mathrm{c}}^{\cal A}(I).
\]
Recall that $\mathrm{MMS}^*(I)=1$ for any instance $I$ in the identical case, and thus 
\begin{equation}\label{prop=mms}
\frac{\mathrm{PROP}^{\cal A}(I)}{\mathrm{PROP}^*(I)} =\frac{\mathrm{MMS}^{\cal A}(I)}{\mathrm{MMS}^*(I)}.
\end{equation}
It follows that the competitive ratio of any online algorithm $H$ with respect to MMS equals that with respect to PROP.

Therefore, finding the optimal approximation guarantee under PROP for goods and chores allocation are equivalent to parallel machine scheduling problems $Pn||C_{\min}$ and $Pn||C_{\max}$, respectively, which have been extensively studied in scheduling literature \cite{DFL82, Graham66}. For online allocation with identical utilities under PROP, results from the online scheduling problems $Pn||C_{\min}$ \cite{Woeginger97} and $Pn||C_{\max}$ \cite{Graham66, FW00, GRT00, CVW94, RC03} can be directly applied. In particular, the normalized case corresponds to the semi-online variant where the total processing time of all jobs is known as a priori. Consequently, we can leverage existing results on corresponding semi-online scheduling problem with known total processing time \cite{TW07, KKST97, AH12, KKG15}. 

Although $\mathrm{EF}^{\cal A}(I)$ can also be interpreted as the ratio between the makespan and the minimum machine load \cite{Wu05}, similar objectives have been less studied in online scheduling. The other fairness notions, such as EFX, EF1, PROPX, PROP1, etc., do not have obvious significance in the context of scheduling and, consequently, have not been investigated. 
Nevertheless, we can still utilize some results obtained from scheduling. For example, it is well known that for any allocation ${\cal A}$,  
\begin{equation}\label{eq:cmax}
	u_{\max}^{\cal A}\geq \max\left\{\frac{1}{n}\sum_{j=1}^mu(e_j),\ \max_{1\leq j \leq m}u(e_j)\right\},	
\end{equation}
and for each $0 \le s \le n-1$, 
\begin{equation}\label{eq:cmin}
	u_{\min}^{\cal A} \le \frac{u(M) - \sum_{k=1}^{s} u(e_{j_k})}{n-s},
\end{equation}
where $e_{j_1}, e_{j_2}, \dots, e_{j_{n-1}}$ are arbitrary $n-1$ items in $M$ \cite{McNaughton59, HT03}.

It is worth noting that in scheduling, the problem where jobs require different processing times on different machines is referred to as unrelated machine scheduling \cite{LST90,BS06}. However, this does not directly correspond to fair division in the general case. This is because, in most cases, determining whether a fairness requirement is satisfied depends on comparing the utilities of different bundles within each agent, rather than comparing the utilities of bundles obtained among different agents.

Moreover, for the general utilities case, \eqref{prop=mms} does not necessarily hold, even if $n=2$ and $\mathrm{MMS}^*=1$. For example, consider an instance of goods allocation with $3$ items. The values of items are  
$$v_1(e_1)=v_2(e_1)=0.5,\ v_1(e_2)=1.4,\ v_2(e_2)=1.3,\ v_1(e_3)=0.1,\ v_2(e_3)=0.2. $$
We can verify that $\text{MMS}_1=0.6$ from allocation ${\cal A}'$ with $A'_1=\{e_1,e_3\}$ and $A'_2=\{e_2\}$, and $\text{MMS}_2=0.7$ from allocation ${\cal A}''$ with $A''_1=\{e_2\}$ and $A''_2=\{e_1,e_3\}$. Both allocations are optimal MMS allocations, and $\text{MMS}^*=1$. ${\cal A}''$ is also an optimal PROP allocation, and $\text{PROP}^*=0.7$. Let ${\cal H}$ with $H_1=\{e_3\}$ and $H_2=\{e_1, e_2\}$ be the allocation produced by some algorithm~$H$. We have $\text{MMS}^{\cal H}=\min\{\frac{v_1(H_1)}{\text{MMS}_1},\frac{v_2(H_2)}{\text{MMS}_2}\}=\frac{0.1}{0.6}=\frac{1}{6}$ and $\text{PROP}^{\cal H}=\min\{\frac{2v_1(H_1)}{v_1(M)},\frac{2v_2(H_2)}{v_2(M)}\}=0.1$. Hence, $\frac{\text{MMS}^{\cal H}}{\text{MMS}^*}=\text{MMS}^{\cal H}\neq \frac{\text{PROP}^{\cal H}}{\text{PROP}^*}$.

\subsection{Relationship Between Goods Allocation and Chores Allocation}

While goods allocation and chore allocation typically exhibit divergent characteristics,  we observe strong connections between the algorithms and their competitive ratios under the EF, EFX, and EF1 notions in the identical case. 

It can be observed from Table \ref{table4} that for any fairness notion $\mathrm{F}\in\{\mathrm{EF},\mathrm{EFX},\mathrm{EF1}\}$, if $v(e)=c(e)$ for any $e\in M$, then
\begin{equation}\label{reciprocal}
	\mathrm{F}^{\cal A}_\mathrm{g}(I)=\frac{1}{\mathrm{F}^{\cal A}_\mathrm{c}(I)}	
\end{equation}
for any ${\cal A}\in \Pi_n(M)$.

\begin{theorem}
	\label{thm:goods-chores-reciprocal}
	For the online fair allocation of goods and chores
	with identical utility functions,
	under any fairness notion $\mathrm{F}\in\{\mathrm{EF},\mathrm{EFX},\mathrm{EF1}\}$,
	the competitive ratio of any algorithm for chores is the reciprocal of that for goods, that is,
	\[
	\mathrm{CR}_{\mathrm{c}}^{\mathrm{F}}(H)
	=
	\frac{1}{\mathrm{CR}_{\mathrm{g}}^{\mathrm{F}}(H)}.
	\]
\end{theorem}

\begin{proof}
	For any instance $I$ of chores allocations, we construct a corresponding instance of goods allocations such that $c(e)=v(e)$ for any $e\in M$. For any fairness notion $\mathrm{F}\in\{\mathrm{EF},\mathrm{EFX},\mathrm{EF1}\}$ and for any algorithm~$H$, 
    $\mathrm{F}^{\cal H}_\mathrm{g}(I)=\frac{1}{\mathrm{F}^{\cal H}_\mathrm{c}(I)}$ by (\ref{reciprocal}). Again by (\ref{reciprocal}), $\mathrm{F}^*_\mathrm{g}(I)=\frac{1}{\mathrm{F}^*_\mathrm{c}(I)}$.
Thus
$$
	\mathrm{CR}_{\mathrm{c}}^{\mathrm{F}}(H)=\sup_{I}\frac{\mathrm{F}^{\cal H}_{\mathrm{c}}(I)}{\mathrm{F}^*_{\text{c}}(I)}=\left( \inf_{I}\frac{\mathrm{F}^{\cal H}_{\mathrm{g}}(I)}{\mathrm{F}^*_{\text{g}}(I)} \right)^{-1}
	\leq \frac{1}{\mathrm{CR}_{\mathrm{g}}^{\mathrm{F}}(H)}.
$$
Similarly, we can also prove that $\mathrm{CR}_{g}^{\mathrm{F}}(H)\leq \frac{1}{\mathrm{CR}_{c}^{\mathrm{F}}(H)}$. Hence, $\mathrm{CR}_{c}^{\mathrm{F}}(H)= \frac{1}{\mathrm{CR}_{g}^{\mathrm{F}}(H)}$.
\end{proof}

As a result, for the case of identical utilities under the EF, EFX, and EF1 notion, we restrict our discussion to goods allocation.

\subsection{Relationship Between Fair Allocation under EF and PROP}

In this subsection, we present the relationship between the competitive ratios of online algorithms for fair allocation with two agents under the EF and PROP notions.

\begin{theorem}\label{thm:cr_prop_ef}
(i) For goods allocation in the normalized setting with two agents, an online algorithm that achieves a competitive ratio of $r_1$ under \emph{PROP} guarantees a competitive ratio of at least $\frac{r_1}{2 - r_1}$ under \emph{EF}.

(ii) For chores allocation in the normalized setting with two agents, an online algorithm that achieves a competitive ratio of $r_2$ under \emph{EF} guarantees a competitive ratio of at most $\frac{2 r_2}{1 + r_2}$ under \emph{PROP}.
\end{theorem}

\begin{proof}
(i) From Table \ref{table4}, we have
\begin{equation}\label{t63}
	\mathrm{EF}_{\mathrm{g}}^{\cal A}(I)=\min\left\{1, \frac{v_1(A_1)}{2-v_1(A_1)}, \frac{v_2(A_2)}{2-v_2(A_2)}\right\}=\frac{\min\{1, v_1(A_1), v_2(A_2)\}}{2-\min\{1, v_1(A_1), v_2(A_2)\}}=\frac{\text{PROP}_{\mathrm{g}}^{\cal A}(I)}{2-\text{PROP}_{\mathrm{g}}^{\cal A}(I)}
\end{equation}
Note that for any $\alpha \le \beta \le 1$,
\begin{equation}\label{eq:goods_cr_prop_ef_2}
	\frac{\frac{\alpha}{2-\alpha}}{\frac{\beta}{2-\beta}}-\frac{\frac{\alpha}{\beta}}{2-\frac{\alpha}{\beta}}=\frac{\alpha}{\beta}\left(\frac{2-\beta}{2-\alpha}-\frac{\beta}{2\beta-\alpha}\right)=\frac{2\alpha(\beta-\alpha)(1-\beta)}{\beta(2-\alpha)(2\beta-\alpha)} \ge 0.
\end{equation}
Suppose that $H$ is an online algorithm for goods allocation with a competitive ratio of $r_1$ under PROP. Since $\text{PROP}_{\mathrm{g}}^{\cal H}(I)\leq \text{PROP}_{\mathrm{g}}^{*}(I)\leq 1$, by \eqref{t63} and \eqref{eq:goods_cr_prop_ef_2}, 
$$
\begin{aligned}
	\text{CR}^{\text{EF}}_{\mathrm{g}}(H)
	&=\inf_I\frac{\text{EF}_{\mathrm{g}}^{\cal H}(I)}{\text{EF}_{\mathrm{g}}^*(I)}
	=\inf_I \frac{\frac{\text{PROP}_{\mathrm{g}}^{\cal H}(I)}{2-\text{PROP}_{\mathrm{g}}^{\cal H}(I)} }{\frac{\text{PROP}_{\mathrm{g}}^{*}(I)}{2-\text{PROP}_{\mathrm{g}}^{*}(I)}}\ge \inf_I \frac{\frac{\text{PROP}_{\mathrm{g}}^{\cal H}(I)}{\text{PROP}_{\mathrm{g}}^*(I)}}{2-\frac{\text{PROP}_{\mathrm{g}}^{\cal H}(I)}{\text{PROP}_{\mathrm{g}}^*(I)}} \\
	&= \frac{\inf_I \frac{\text{PROP}_{\mathrm{g}}^{\cal H}(I)}{\text{PROP}_{\mathrm{g}}^*(I)}}{2-\inf_I\frac{\text{PROP}_{\mathrm{g}}^{\cal H}(I)}{\text{PROP}_{\mathrm{g}}^*(I)}}=\frac{\text{CR}^{\text{PROP}}_{\mathrm{g}}(H)}{2-\text{CR}^{\text{PROP}}_{\mathrm{g}}(H)}=\frac{r_1}{2-r_1}.
\end{aligned}
$$
Hence, $H$ is an online algorithm with a competitive ratio at least $\frac{r_1}{2-r_1}$ under EF. 

(ii) From Table \ref{table4}, we have 
\begin{equation*}
	\mathrm{EF}_{\mathrm{c}}^{\cal A}(I)=\max\left\{1, \frac{c_1(A_1)}{2-c_1(A_1)}, \frac{c_2(A_2)}{2-c_2(A_2)}\right\}=\frac{\max\{1, c_1(A_1), c_2(A_2)\}}{2-\max\{1, c_1(A_1), c_2(A_2)\}}=\frac{\text{PROP}_{\mathrm{c}}^{\cal A}(I)}{2-\text{PROP}_{\mathrm{c}}^{\cal A}(I)}.
\end{equation*}
Thus
\begin{equation}\label{t64}
\text{PROP}_{\mathrm{c}}^{\cal A}(I)=\frac{2\mathrm{EF}_{\mathrm{c}}^{\cal A}(I)}{1+\mathrm{EF}_{\mathrm{c}}^{\cal A}(I)}.	
\end{equation}
Note that for any $\alpha \ge \beta \ge 1$,
\begin{equation}\label{eq:goods_cr_prop_ef_3}
    \frac{\frac{2\alpha}{1+\alpha}}{\frac{2\beta}{1+\beta}}-\frac{2\cdot\frac{\alpha}{\beta}}{1+\frac{\alpha}{\beta}}
    =\frac{\alpha}{\beta}\left(\frac{1+\beta}{1+\alpha}-\frac{2\beta}{\alpha+\beta}\right)
    =\frac{\alpha(\alpha-\beta)(1-\beta)}{\beta(1+\alpha)(\alpha+\beta)} \le 0.
    \end{equation}
Suppose that $H$ is an online algorithm for chores allocation with a competitive ratio of $r_2$ under EF. Since $\text{EF}_{\mathrm{c}}^{\cal H}(I)\geq \text{EF}_{\mathrm{c}}^{*}(I)\geq 1$, by \eqref{t64} and \eqref{eq:goods_cr_prop_ef_3}, 
$$
    \begin{aligned}
    \text{CR}^{\text{PROP}}_{\mathrm{c}}(H)
    &=\sup_I\frac{\text{PROP}_{\mathrm{c}}^{\cal H}(I)}{\text{PROP}_{\mathrm{c}}^*(I)}
    =\sup_I \frac{\frac{2\text{EF}_{\mathrm{c}}^{\cal H}(I)}{1+\text{EF}_{\mathrm{c}}^{\cal H}(I)} }{\frac{\text{EF}_{\mathrm{c}}^{*}(I)}{1+\text{EF}_{\mathrm{c}}^{*}(I)}}\le \sup_I \frac{2\cdot\frac{\text{EF}_{\mathrm{c}}^{\cal H}(I)}{\text{EF}_{\mathrm{c}}^*(I)}}{1+\frac{\text{EF}_{\mathrm{c}}^{\cal H}(I)}{\text{EF}_{\mathrm{c}}^*(I)}} \\
    &= \frac{2\sup_I \frac{\text{EF}_{\mathrm{c}}^{\cal H}(I)}{\text{EF}_{\mathrm{c}}^*(I)}}{1+\sup_I\frac{\text{EF}_{\mathrm{c}}^{\cal H}(I)}{\text{EF}_{\mathrm{c}}^*(I)}}=\frac{2\text{CR}^{\text{EF}}_{\mathrm{c}}(H)}{1+\text{CR}^{\text{EF}}_{\mathrm{c}}(H)}=\frac{2r_2}{1+r_2}.
    \end{aligned}
$$
Hence, $H$ is an online algorithm with a competitive ratio at most $\frac{2r_2}{1+r_2}$ under PROP. 

\end{proof}

However, the converse of Theorem \ref{thm:cr_prop_ef} may not hold, even in the identical utilities case. For example, consider an instance of goods allocation with $3$ items. The values of items are $v(e_1)=1.1$, $v(e_2)=0.5$ and $v(e_3)=0.4$. 
We can verify that $\text{EF}^*=\frac{9}{11}$ and $\text{PROP}^*=\frac{9}{10}$ from allocation ${\cal A}^*$ with $A^*_1=\{e_1\}$ and $A^*_2=\{e_2, e_3\}$. Let ${\cal H}$ with $H_1=\{e_3\}$ and $H_2=\{e_1, e_2\}$ be the allocation produced by some algorithm~$H$. We have $\text{EF}^{\cal H}=\frac{1}{4}$ and $\text{PROP}^{\cal H}=\frac{2}{5}$. Hence, $\frac{\text{EF}^{\cal H}}{\text{EF}^*}=\frac{11}{36}$ and 
$$\frac{\text{PROP}^{\cal H}}{\text{PROP}^*}=\frac{4}{9}<\frac{22}{47}=\frac{2\frac{\text{EF}^{\cal H}}{\text{EF}^*}}{1+\frac{\text{EF}^{\cal H}}{\text{EF}^*}}.$$

\section{Conclusion and Future Work}

In this paper, we revisit online fair division
from the perspective of competitive analysis with respect to approximation guarantees.
We refine the definition of competitive ratio for online fair division
to align it with the classical definition in online algorithms,
while remaining consistent with existing results in the online fair division literature.
This refined framework enables a systematic and comparable study
of approximation guarantees across different fairness notions.
To the best of our knowledge,
our results provide the first comprehensive characterization
of competitive ratios for online fair division,
covering both goods and chores within a unified framework.

Several directions remain open for future work. A natural step is to bridge the existing gaps where current upper and lower bounds do not match. It would also be worthwhile to explore alternative fairness notions, such as EFk \cite{CKMPSW19}, EFkX \cite{ARS22}, Avg-EFX \cite{BGGS21}, PROPm \cite{BGGS21}, and PROPavg \cite{KM25}. Furthermore, these problems can be extended to the setting of weighted fair division.

In this paper, we investigate the normalized setting in which the total utility of all items is known in advance. This framework can be viewed as a semi-online model, a paradigm that has been extensively studied in the scheduling literature. While various semi-online models incorporating different types of partial information have been proposed, many of these can be naturally adapted to the context of online fair division. 

Finally, online algorithms with predictions, which bridge machine learning and online problems, constitute an active area of research. In fair division with predictions, predicted values for item utilities can be obtained, and these predictions are subject to quantifiable errors. The goal is to leverage these predictions to achieve an allocation  that outperforms that obtained without any utility information.


\begin{thebibliography}{2}

\bibitem{AG24}
Akrami H, Garg J: Breaking the 3/4 barrier for approximate maximin share. Proceedings of the 35th Annual ACM-SIAM Symposium on Discrete Algorithms, 74-91, 2024.

\bibitem{ARS22}
Akrami H, Rezvan R, Seddighin M:
An EF2X allocation protocol for restricted additive valuations.
Proceedings of the 31st International Joint Conference on Artificial Intelligence, 17-23, 2022.

\bibitem{AH12}
Albers S, Hellwig M: 
Semi-online scheduling revisited. 
Theoretical Computer Science, 443, 1-9, 2012.

\bibitem{AAG15}
Aleksandrov M, Aziz H, Gaspers S, Walsh T: 
Online fair division: Analysing a food bank problem.
Proceedings of the 24th International Joint Conference on Artificial Intelligence, 2540-2546, 2015.

\bibitem{AW20}
Aleksandrov M, Walsh T: 
Online fair division: A survey.
Proceedings of the 34th AAAI Conference on Artificial Intelligence, 13557-13562, 2020.


\bibitem{ABM18}
Amanatidis G, Birmpas G, Markakis E:
Comparing approximate relaxations of envy-freeness.
Proceedings of the 27th International Joint Conference on Artificial Intelligence, 42-48, 2018.

\bibitem{AFS24}
Amanatidis G, Filos-Ratsikas A, Sgouritsa A:
Pushing the frontier on approximate EFX allocations.
Proceedings of the 25th ACM Conference on Economics and Computation, 1268-1286, 2024.

\bibitem{ALM25}
Amanatidis G, Lolos A, Markakis E, Turmel V: 
Online fair division for personalized 2-value instances.
Proceeding of the 18th International Symposium on Algorithmic Game Theory, Lecture Notes in Computer Science, 15953, 209-227, Springer, 2025.

\bibitem{AMN17}
Amanatidis G, Markakis E, Nikzad A, Saberi A:
Approximation algorithms for computing maximin share allocations. 
ACM Transactions on Algorithms, 13(4), Article 52, 2017.

\bibitem{AMN20}
Amanatidis G, Markakis E, Ntokos A:
Multiple birds with one stone: Beating 1/2 for EFX and GMMS via envy cycle elimination. 
Theoretical Computer Science, 841, 94-109, 2020.

\bibitem{ALM24}
Aziz H, Li B, Moulin H, Wu X, Zhu X:
Almost proportional allocations of indivisible chores: Computation, approximation and efficiency. 
Artificial Intelligence, 331, 104118, 2024

\bibitem{BGG22}
Banerjee S, Gkatzelis V, Gorokh A, Jin B:
Online Nash social welfare maximization with predictions.
Proceedings of the 33th Annual ACM-SIAM Symposium on Discrete Algorithms, 1-19, 2022.


\bibitem{BGGS21}
Baklanov A, Garimidi P, Gkatzelis V, Schoepflin D:
Achieving proportionality up to the maximin item with indivisible goods.
Proceedings of the 35th AAAI Conference on Artificial Intelligence, 5143-5150, 2021.

\bibitem{BS06}
Bansal N, Sviridenko M: 
The Santa Claus problem.
Proceedings of the 38th Annual ACM Symposium on Theory of Computing, 31-40, 2006.

\bibitem{BKP18}
Benade G, Kazachkov A M, Procaccia A D, Psomas C A: 
How to make envy vanish over time.
Proceedings of the 19th ACM Conference on Economics and Computation, 593-610, 2018.

\bibitem{BE05}
Borodin A, El-Yaniv R: 
Online Computation and Competitive Analysis. 
Cambridge University Press, 2005.

\bibitem{Budish11}
Budish E: 
The combinatorial assignment problem: Approximate competitive equilibrium from equal incomes. 
Journal of Political Economy, 119(6), 1061-1103, 2011.


\bibitem{CKMPSW19}
Caragiannis I, Kurokawa D, Moulin H, Procaccia AD, Shah N, Wang J:
The unreasonable fairness of maximum Nash welfare.
ACM Transactions on Economics and Computation, 7(3), Article 12, 2019.

\bibitem{CGM24}
Chaudhury B R, Garg J, Mehlhorn K: 
EFX exists for three agents. 
Journal of the ACM, 71(1), Article 4, 2024.

\bibitem{CVW94}
Chen B, van Vliet A, Woeginger G J: 
New lower and upper bounds for on-line scheduling. 
Operations Research Letters, 16(4), 221-230, 1994.

\bibitem{CFKNPT26}
Choo D, Fu W, Khu D, Neoh TY, Poon TY, Teh N: 
Approximate proportionality in online fair division.
Proceedings of the 43th International Conference on Machine Learning, 2026. Also on arXiv: 2508.03253. 

\bibitem{CFS17} 
Conitzer V, Freeman R, Shah N: 
Fair public decision making. 
Proceedings of the 18th ACM Conference on Economics and Computation, 629-646, 2017.

\bibitem{DFL82}
Deuermeyer B L, Friesen D K, Langston M A: 
Scheduling to maximize the minimum processor finish time in a multiprocessor system. 
SIAM Journal on Algebraic Discrete Methods, 3(2), 190-196, 1982.

\bibitem{ELL25}
Elkind E, Lam A, Latifian M, Neoh T Y, Teh N: 
Temporal fair division of indivisible items.
Proceedings of the 24th International Conference on Autonomous Agents and Multiagent Systems, 676-685, 2025.

\bibitem{Epstein18}
Epstein L: 
A survey on makespan minimization in semi-online environments. 
Journal of Scheduling, 21(3), 269-284, 2018.

\bibitem{FKT89} 
Faigle U, Kern W, Turan G: 
On the performance of on-line algorithms for partition problems. 
Acta Cybernetica, 9(2), 107-119, 1989.

\bibitem{FN22}
Feige U, Norkin A:
Improved maximin fair allocation of indivisible items to three agents. 
arXiv:2205.05363, 2022.

\bibitem{FW00}
Fleischer R, Wahl M: 
On‐line scheduling revisited. 
Journal of Scheduling, 3(6), 343-353, 2000.

\bibitem{Foley66}
Foley D K: 
Resource allocation and the public sector. 
Yale Economic Essays, 7, 45-98, 1966.

\bibitem{GMQ25}
Garg J, Murhekar A, Qin J: 
Constant-factor EFX exists for chores.
Proceedings of the 57th Annual ACM Symposium on Theory of Computing, 1580-1589, 2025.

\bibitem{GHSSY18}
Ghodsi M, HajiAghayi MT, Seddighin M, Seddighin S, Yami H:
Fair allocation of indivisible goods: Improvements. 
Mathematics of Operations Research, 2021, 46(3):1038-1053. Also on arXiv:1704.00222.

\bibitem{GLT25}
Gong M, Lin G, Tan Z:
Semi-online multiprocessor scheduling with known largest job processing time.
Proceedings of the 17th International Conference on Combinatorial Optimization and Applications, Lecture Notes in Computer Science, 15434, 213-225, Springer, 2025.

\bibitem{GRT00}
Gormley T, Reingold N, Torng E, Westbrook J: 
Generating adversaries for request-answer games.
Proceedings of the 11th annual ACM-SIAM Symposium on Discrete Algorithms, 564-565, 2000.

\bibitem{Graham66}
Graham R L: 
Bounds for certain multiprocessing anomalies. 
Bell System Technical Journal, 45(9), 1563-1581, 1966.

\bibitem{HT03}
He Y, Tan Z: 
Randomized on-line and semi-on-line scheduling on identical machines.
Asia-Pacific Journal of Operational Research, 20(1), 2003.


\bibitem{HS23}
Huang X, Segal-Halevi E:
A reduction from chores allocation to job scheduling.
Proceedings of the 24th ACM Conference on Economics and Computation, 908-908, 2023.


\bibitem{KKST97}
Kellerer H, Kotov V, Speranza M G, Tuza Z: 
Semi on-line algorithms for the partition problem. 
Operations Research Letters, 21(5), 235-242, 1997.

\bibitem{KKG15}
Kellerer H, Kotov V, Gabay M: 
An efficient algorithm for semi-online multiprocessor scheduling with given total processing time. 
Journal of Scheduling, 18(6), 623-630, 2015.

\bibitem{KM25}
Kobayashi Y, Mahara R: 
Proportional allocation of indivisible goods up to the least valued good on average. 
SIAM Journal on Discrete Mathematics, 39, 533-549, 2025.

\bibitem{LST90}
Lenstra J K, Shmoys D B, Tardos É: 
Approximation algorithms for scheduling unrelated parallel machines. 
Mathematical Programming, 46(1), 259-271, 1990.

\bibitem{LMM04}
Lipton R J, Markakis E, Mossel E, Saberi A: 
On approximately fair allocations of indivisible goods.
Proceedings of the 5th ACM Conference on Electronic Commerce, 125-131, 2004.

\bibitem{McNaughton59}
McNaughton R: 
Scheduling with deadlines and loss functions. 
Management Science, 6(1), 1-12, 1959.

\bibitem{MP25}
Melissourgos T, Protopapas N: 
Online EFX allocations with predictions. 
Proceedings of the 25th International Conference on Autonomous Agents and Multiagent Systems, 3519–3521, 2026.

\bibitem{Moulin19} 
H. Moulin: 
Fair division in the internet age.
Annual Review of Economics, 11, 407-441, 2019.

\bibitem{NPT25}
Neoh T Y, Peters J, Teh N: 
Online fair division with additional information. 
Proceedings of the 43th International Conference on Machine Learning, 2026. Also on arXiv:2505.24503.

\bibitem{PR20}
Plaut B, Roughgarden T: 
Almost envy-freeness with general valuations. 
SIAM Journal on Discrete Mathematics, 34(2), 1039-1068, 2020.

\bibitem{PW14}
Procaccia A D, Wang J: 
Fair enough: Guaranteeing approximate maximin shares.
Proceedings of the 15th ACM Conference on Economics and Computation, 675-692, 2014.

\bibitem{RC03}
Rudin III J F, Chandrasekaran R: 
Improved bounds for the online scheduling problem. 
SIAM Journal on Computing, 32(3), 717-735, 2003.

\bibitem{Steihaus48} 
Steihaus H: 
The problem of fair division. 
Econometrica, 16, 101–104, 1948.

\bibitem{STW25}
Song J, Tao B, Wang W, Zhang Y: 
Online MMS allocation for chores. 
arXiv:2507.14039, 2025.

\bibitem{TH07}
Tan Z, He Y:
Semi-online scheduling problems on two identical machines with inexact partial information.
Theoretical Computer Science, 377, 110-125, 2007.

\bibitem{TW07}
Tan Z, Wu Y: 
Optimal semi-online algorithms for machine covering. 
Theoretical Computer Science, 372(1), 69-80, 2007.


\bibitem{TZ13}
Tan Z, Zhang A: 
Online and semi-online scheduling.
Handbook of Combinatorial Optimization (2nd Edition). Springer, 2191-2252, 2013.


\bibitem{WW26}
Wang Y, Wei T: 
Online fair allocations with binary valuations and beyond.
Proceedings of the 40th AAAI Conference on Artificial Intelligence, 17267-17275, 2026.

\bibitem{Woeginger97}
Woeginger G J:
A polynomial-time approximation scheme for maximizing the minimum machine completion time. 
Operations Research Letters, 20(4), 149-154, 1997.


\bibitem{Wu05}
Wu B Y:
An analysis of the LPT algorithm for the max–min and the min–ratio partition problems. 
Theoretical Computer Science, 349(3), 407-419, 2005.


\bibitem{ZBW23}
Zhou S, Bai R, Wu X: 
Multi-agent online scheduling: MMS allocations for indivisible items.
Proceedings of the 40th International Conference on Machine Learning, 42506-4251, 2023.
\end{thebibliography}
\end{document}